\documentclass[aps,pra,twocolumn,amsfonts, amssymb, amsmath, amsthm,groupedaddress]{revtex4-1}
\usepackage{url}
\usepackage{color}
\usepackage{graphicx}
\usepackage{amsthm}
\newtheorem{definition}{Definition}
\newtheorem{proposition}{Proposition}
\newtheorem*{remark}{Remark}
\DeclareMathOperator{\Tr}{\operatorname{Tr}}

\usepackage{relsize}

\begin{document}

\title{Qubit positive-operator-valued measurements by destructive weak measurements}

\author{Yi-Hsiang Chen}
\affiliation{Department of Physics and Astronomy, University of Southern California, Los Angeles, California 90089, USA}
\author{Todd A. Brun}
\affiliation{Communication Sciences Institute, University of Southern California, Los Angeles, California 90089, USA}

\date{\today}

\begin{abstract}
Many quantum measurements, such as photodetection, can be destructive. In photodetection, when the detector ``clicks'' a photon has been absorbed and destroyed. Yet the lack of a click also gives information about the presence or absence of a photon. In monitoring the emission of photons from a source, one decomposes the strong measurement into a series of weak measurements, which describe the evolution of the state during the measurement process. Motivated by this example of destructive photon detection, a simple model of destructive weak measurements using qubits was studied in earlier work by the authors. It has shown that the model can achieve any positive-operator-valued measurement (POVM) with commuting POVM elements including projective measurements. In this paper, we use a different approach for decomposing \textit{any} POVM into a series of weak measurements. The process involves three steps: randomly choose, with certain probabilities, a set of linearly independent POVM elements; perform that POVM by a series of destructive weak measurements; output the result of the series of weak measurements with certain probabilities. The probabilities of the outcomes from this process agree with those from the original POVM, and hence this model of destructive weak measurements can perform any qubit POVM. 
\end{abstract}


\maketitle

\section{Introduction}\label{intro}
Quantum trajectory theory has been widely discussed and studied for years. Many essential aspects of its physics, such as the stochastic evolution equations \cite{Gisin_1984,Gisin_1993,Diosi1988,Diosi_1988} arising from the interaction with the continuously measured environment leading to diffusive or jump-like behavior for the system, are well-illustrated in \cite{Todd:2002,doi:10.1080/00107510601101934}. Quantum trajectories are closely related to continuous measurements, which can be thought of as a limit of repeated weak measurements.  Continuous measurements have been studied and applied to various fields, including quantum optics \cite{Wiseman_1996}, superconducting qubits \cite{Roch2014,Devoret1169,Tan2015}, and quantum state tomography \cite{Shojaee2018,Wu2018,Areeya2018}. A feature worth pointing out in continuous measurements is the use of an adaptive procedure that updates the current status for the next step of measurement. This feedback technique has both experimental advantage for practical designs \cite{Wiseman:1995aa,Wiseman_1998,Sayrin2011,Campagne-Ibarcq2013,Vijay2012} and theoretical importance for the ability to realize more general classes of quantum measurements \cite{Oreshkov:2005aa,Varbanov:2007aa,Florjanczyk:2014aa,Florjanczyk:2015aa}. In quantum optics, for example, the phase measurement of a single mode can be improved by feedback loops during the measurement process \cite{Wiseman:1995aa,Wiseman_1998}, and it was also used in the preparation and stabilization of photon number states \cite{Sayrin2011}. Such an adaptive measuring scheme is also essential in the active field of superconducting qubits \cite{Vijay2012,ClarkeWilhelm2008,Campagne-Ibarcq2013}. In the theory of continuous measurements, the work in \cite{Oreshkov:2005aa,Varbanov:2007aa} on decomposing quantum measurements into a series of weak measurements requires such a continuously updating process: the measurement operators at each time step are updated by the outcomes of the previous weak measurements. It is also well-known that any generalized measurement for a system is equivalent to letting the system interact with an ancilla followed by a projective measurement on the ancilla. In \cite{Florjanczyk:2014aa,Florjanczyk:2015aa}, it is shown that what classes of generalized measurements can be achieved by tuning the probe state or the interaction Hamiltonian during the continuous measurement process.

Since a generalized measurement is represented by a set of operators $M_k$'s such that $\sum_k M_k^{\dagger}M_k=I$, each $M_k$ is naturally decomposed into a unitary multiplying by a positive semi-definite operator by the polar decomposition, i.e., $M_k=U_k \sqrt{E_k}$. The unitary operators $U_k$ do not affect the probabilities of getting the outcomes, but give an evolution conditional on the measurement outcome. The positive semi-definite operators $E_k$ extract the information of the positive-operator-valued measurement (POVM) associated with this generalized measurement. In the light of this observation, we look at the case of ``destructive" measurements where the system is projected onto a certain state after the measurement, while the probabilities of the outcomes reflect a certain POVM. A simple example of this type of measurements is the two-outcome qubit measurement consisting of $M_1=|0\rangle\langle+|$ and $M_2=|0\rangle\langle -|$. It projects the system onto the $|0\rangle$ state for both outcomes while performing a POVM in the x-basis. A physical example of a destructive measurement is photodetection. In photodetection, the observer waits for photons to be absorbed by the detector. After a long monitoring time, all photons are consumed and the source system is projected onto the vacuum state. If we continuously monitor this detection process, the outcomes give a stochastic evolution for the process of photodetection. 

Motivated by this example, we developed a model of destructive weak measurements for qubits in \cite{Chen:1}. It consists of a weak unitary swap between the system and the ancilla qubit, in the initial state $|0\rangle $, followed by a projective measurement on the ancilla qubit. The weak swap between the system and the $|0\rangle$ state is the same for each step and the cumulative effect of the swap leads to a projection onto $|0\rangle$ for the system at long times. This is analogous to photodetection, in which photons must be absorbed to be detected: the source system is projected onto the vacuum state---the analog of $|0\rangle$---after all photons are consumed. Note that the destructiveness in this toy model does not come from the projective measurement on the ancilla, but the constant weak swap interaction between the system and the ancilla qubits and the fact that the ancillas begin in the state $|0\rangle$ at each step. Although the model is destructive, it can achieve any POVM with commuting POVM elements, including all projective measurements on a qubit. However, general POVMs can not be done by this method. In this paper, we develop another 
method, the random walk in a simplex, to achieve an arbitrary POVM with any number of outcomes. 

The paper is structured as follows. In Sec.~\ref{0}, we briefly introduce the model of destructive weak measurements in \cite{Chen:1}. In Sec.~\ref{1}, we present some general properties for a POVM. They are essential to prove the POVM achievability by the model of destructive weak measurement. In Sec.~\ref{2}, we describe the method of random walks in a simplex and show how it fits our model of destructive weak measurements. In Sec.~\ref{3}, we show that the probabilities of getting the outcomes approach the probabilities from the POVM, as expected.

\section{The model of destructive weak measurements}
\label{0}
Recall the model of destructive weak measurements in \cite{Chen:1}: we couple a qubit system $|\psi\rangle\langle \psi|$ to an ancilla in state $|0\rangle$ by a weak swap defined as
\begin{equation}
U=I \cos \phi -i S \sin \phi, 
\end{equation}
where $S$ is the swap operator that switches the system and the ancilla qubit, i.e., $ S|\psi \rangle \otimes |0\rangle=|0\rangle\otimes |\psi\rangle$.
The parameter $\phi$ is assumed to be a small constant. The joint system becomes 
\begin{equation}
U|\psi\rangle\otimes|0\rangle=|\psi\rangle \otimes|0\rangle\cos\phi-i |0\rangle \otimes|\psi\rangle\sin\phi
\end{equation}
after the unitary operation. We perform a POVM on the ancilla with elements of the form
\begin{equation}
\hat{P}_k=s_k |e_k\rangle\langle e_k|,
\end{equation}
where
\begin{equation}
 \langle e_k|e_k\rangle=1,  \ \ 0<s_k\leq 1, \ \ \sum^n_{k=1}\hat{P}_k=I.
\end{equation}
The outcome $\hat{P}_k$ from the ancilla gives a measurement operator
\begin{equation}
\label{meas.op.}
M_k=\sqrt{s_k}\left(\langle e_{k} | 0\rangle I\cos \phi   - i |0\rangle\langle e_{k}|\sin\phi  \right)
\end{equation}
on the system, and $\sum_{k}M^{\dagger}_{k}M_{k}=I$.

\section{General Properties of POVMs}
\label{1}
\begin{definition}[$Linearly$ $independent$ $POVM$]
We say that a POVM, $\{E_i\}^n_{i=1}$, is linearly independent if $\sum^n_{i=1}E_i=I$, where $E_i$'s are positive-semidefinite operators, and $\sum^n_{i=1}c_iE_i=0$ has no nontrivial solutions for the $c_i$'s. We denote such a POVM as a linearly independent POVM (LIPOVM).
\end{definition}

A POVM, in general, can have any number of outcomes. However, it can be decomposed into a linear combination of linearly independent POVMs in the following sense. 
\begin{proposition}
Performing a POVM is equivalent to randomly choosing, with certain probabilities, to perform one LIPOVM from a collection of LIPOVMs.
\end{proposition}
We provide the following decomposing process to illustrate Proposition 1.

Consider a POVM, $\{E_i\}^n_{i=1}$, with linearly dependent POVM elements. The $E_i$'s are positive-semidefinite operators with $\sum^n_{i=1}E_i=I$, and $\sum^n_{i=1}c_iE_i=0$ has a nontrivial solution for the $c_i$'s. We use a labeling convention such that $c_1\geq c_2\geq \cdots \geq c_n$, where $c_1,\cdots, c_m\geq 0$ and $c_{m+1},\cdots, c_n < 0$. Note that $c_1>0$ and $c_n<0$ are always true, since otherwise they contradict the positivity of the $E_i$'s. We can construct two POVMs,
\begin{align}
E^A_j\equiv\frac{c_1-c_j}{c_1}E_j,&\ \ j=2,\cdots ,n \\
E^B_j\equiv\frac{c_n-c_j}{c_n}E_j,&\ \ j=1,\cdots ,n-1.
\end{align}
One can easily check that $\sum^n_{j=2}E^A_j=I$ and $\sum^{n-1}_{j=1}E^B_j=I$. We then perform $\{E^A_j\}^{n}_{j=2}$ with probability $P_A=c_1/(c_1-c_n)$ and perform $\{E^B_j\}_{j=1}^{n-1}$ with probability $P_B=-c_n/(c_1-c_n)$. The probability of obtaining each outcome can be checked to agree with the original value $\Tr E_i \rho$ for a given state $\rho$. For outcome $n$, 
\begin{equation}
P_A\Tr \left[E^A_n \rho\right]= \frac{c_1}{c_1-c_n}\Tr\left[ \frac{c_1-c_n}{c_1} E_n\rho\right]= \Tr \left[E_n\rho\right].
\end{equation}
For outcome 1,
\begin{equation}
P_B\Tr\left[ E^B_1 \rho\right]= \frac{-c_n}{c_1-c_n}\Tr \left[\frac{c_n-c_1}{c_n} E_1\rho\right]= \Tr \left[E_1\rho\right]. 
\end{equation}
For the outcomes $i$, $2\leq i\leq n-1$,
\begin{align}
&P_A\Tr \left[E^A_i \rho\right] +P_B\Tr\left[ E^B_i \rho\right] =   \nonumber \\
&\frac{c_1}{c_1-c_n} \Tr\left[ \frac{c_1-c_i}{c_1}E_i \rho\right] +\frac{-c_n}{c_1-c_n} \Tr \left[\frac{c_n-c_i}{c_n} E_i\rho\right]  \nonumber\\
&  =\Tr \left[E_i \rho\right].
\end{align}
This process reduces an $n$-outcome POVM into two $(n-1)$-outcome POVMs with pre-processing. If either of $\{E^A_j\}^n_{j=2}$ and $\{E^B_j\}^{n-1}_{j=1}$ is still linearly dependent, one can repeat the same process to reduce the number of outcomes until all POVMs are linearly independent. Hence, an $n$-outcome linearly dependent POVM is equivalent to randomly choosing from a set of LIPOVMs with certain probabilities. Since a qubit POVM has elements that are 2 by 2 matrices, any qubit POVM with more than 4 outcomes can always be decomposed into 4-or-fewer-outcome LIPOVMs by pre-processing.

\begin{definition}[$Projective$ $POVM$]
We call a POVM with each element proportional to a projector a projective POVM (PPOVM), i.e., all $\{E_i\}_{i=1}^n$ have the form $E_i=a_i |i\rangle\langle i|$, where the $a_i$'s are positive.
\end{definition}

\begin{proposition}
Performing any qubit POVM is equivalent to performing a projective POVM followed by outputting an outcome with a probability that depends on the result of that projective POVM.
\end{proposition}

\begin{proof}
Given a qubit POVM $\{E_i\}^n_{i=1}$, each $E_i$ can be written as 
\begin{equation}
E_i=a_i|a_i\rangle\langle a_i| + b_i|b_i\rangle\langle b_i|, 
\end{equation}
where $a_i,b_i$ and $|a_i\rangle,|b_i\rangle$ are the eigenvalues and eigenvectors of $E_i$. Let $a_i\geq b_i$ for all $i$. We can rewrite each $E_i$ as
\begin{equation}
E_i=(a_i-b_i)|a_i\rangle\langle a_i|+ b_i I,
\end{equation}
since $|a_i\rangle\langle a_i|+|b_i\rangle\langle b_i|=I$ for each $i$. Because $\sum^n_{i=1}E_i=I$, we have
\begin{equation}
I=\displaystyle\sum^n_{i=1}\frac{a_i-b_i}{\left(1-\sum^n_{j=1}b_j\right)}|a_i\rangle\langle a_i|.
\end{equation}
Define a PPOVM $\{\tilde{P}_i\}^n_{i=1}$ as
\begin{equation}
\tilde{P}_i\equiv \frac{a_i-b_i}{\left(1-\sum^n_{j=1}b_j\right)}|a_i\rangle\langle a_i|.
\end{equation}
We can rewrite $E_i$ in terms of the $\tilde{P}_i$'s; i.e.,
\begin{align}
\label{POVMtrans}
E_i &= (a_i-b_i)|a_i\rangle\langle a_i| + b_i \displaystyle\sum^n_{k=1}\frac{a_k-b_k}{\left(1-\sum^n_{j=1}b_j\right)}|a_k\rangle\langle a_k| \nonumber \\
&= \displaystyle\sum^n_{k=1}\left[\delta_{ik}(a_k-b_k)+b_i \frac{a_k-b_k}{\left(1-\sum^n_{j=1}b_j\right)}\right]|a_k\rangle\langle a_k| \nonumber \\
&= \displaystyle\sum^n_{k=1}\left[\delta_{ik}\left(1-\sum^n_{j=1}b_j\right)+b_i \right] \tilde{P}_k \equiv  \displaystyle\sum^n_{k=1} p(i|k)  \tilde{P}_k,
\end{align}
where $\delta_{ik}$ is the Kronecker delta and 
\begin{equation}
\label{conditionalprob}
p(i|k)\equiv \delta_{ik}\left(1-\sum^n_{j=1}b_j\right)+b_i .
\end{equation} 
One can quickly check that $\sum^n_{i=1}p(i|k)=1$ for any outcome $k$. Hence, the function $p(i|k)$ can be interpreted as a conditional probability, in the sense that we first perform $\{\tilde{P}_i\}^n_{i=1}$ and then output the final outcome $i$ with probability $p(i|k)$ when we got result $\tilde{P}_k$ from $\{\tilde{P}_i\}^n_{i=1}$. Therefore, performing a qubit POVM $\{E_i\}^{n}_{i=1}$ is equivalent to performing a projective POVM $\{\tilde{P}_i\}^n_{i=1}$ followed by post-processing as defined in Eq. (\ref{conditionalprob}).
\end{proof}
\begin{remark}
If the qubit POVM $\{E_i\}^{n}_{i=1}$ is linearly independent, so is $\{\tilde{P}_i\}^n_{i=1}$. 
\end{remark}

This is because the transformation in Eq. (\ref{POVMtrans}) has an inverse. That is 
\begin{equation}
\tilde{P}_i=\displaystyle\sum^n_{j=1}p^{-1}_{ij}E_j,
\end{equation}
where $p^{-1}_{ij}=(\delta_{ij}-b_i)/(1-\sum^{n}_{k=1}b_k)$. If the set of vectors $\{E_i\}^{n}_{i=1}$ is linearly independent, then $\{\tilde{P}_i\}^n_{i=1}$ will be linearly independent as well.

Combining the above properties, we have the following result: any qubit POVM is equivalent to pre- and post-processing steps with a set of linearly independent and projective POVMs. That is, for a given qubit POVM, one uses a random number generator to choose, with probabilities given in Proposition 1, a linearly independent POVM, and uses Proposition 2 to find and perform the projective POVM corresponding to that linearly independent POVM; finally, one outputs the final outcome with probability $p(\cdot|\cdot)$, defined in Eq. (\ref{conditionalprob}). The probability of outputting the final outcome $i$ will be equal to $\Tr E_i\rho $, which is the probability of getting outcome $i$ from the original POVM. 

Note that the only actual measurement made in this process is the linearly independent and projective POVM (LIPPOVM). Hence, for a qubit system, if one can perform any LIPPOVM then all POVMs can be accomplished. We claim that the model of lossy weak measurements \cite{Chen:1} can achieve any qubit LIPPOVM and hence any qubit POVM. The following sections will verify the assertion.

\section{The achievability of POVMs by destructive weak measurements}
\label{2}

\subsection{POVM random walk in a simplex}
In earlier work \cite{Varbanov:2007aa}, it was shown that a generalized measurement can be decomposed into a series of stochastic processes such that the trajectory corresponds to a random walk in a simplex. Here we use a similar approach to parametrize the measurement operator during the sequence of weak measurements. The intuition is to map the current POVM element at each step of a sequence of weak measurements into a convex combination of the original POVM elements. We will show that, for a linearly independent POVM, the conditions to perform a measurement where each outcome corresponds to a vector in the simplex can be satisfied. 

To decompose a linearly independent POVM into stochastic processes of weak measurements, we use a family of measurement operators $M(\vec{x})$ of the form
\begin{equation}
\label{operatorform}
M(\vec{x})=f(\vec{x})U(\vec{x})\sqrt{\sum^{n}_{i=1}x_i E_i}, 
\end{equation}
where $U(\vec{x})$ is a unitary operator, $f(\vec{x})$ is a normalization factor, and 
\begin{equation}
\vec{x}\in \triangle^{n-1}=\left\{(x_1,\cdots ,x_n)|\sum^n_{i=1}x_i=1, \ x_i\geq 0 \ \forall i\right\}.
\end{equation}
The form of Eq.~(\ref{operatorform}) is fairly general because any matrix allows for a polar decomposition, and Eq.~(\ref{operatorform}) only requires the positive-semidefinite matrix to be the square root of a linear combination of the POVM elements. The vector $\vec{x}$ characterizes the position of the random walk in the $(n-1)$-simplex, $\triangle^{n-1}$. Note that for a given position $\vec{x}$ in the walk, 
\begin{equation}
\label{POVMelem}
M^{\dagger}(\vec{x})M(\vec{x})\propto \sum^n_{i=1}x_i E_i.
\end{equation}
This form characterizes how the POVM element of the weak measurement evolves as a convex combination of the original POVM elements.
It will be shown in Sec.~\ref{3} that $\vec{x}$ will approach one of the vertices of $\triangle^{n-1}$ after many measurements: $\vec{x}\to(\cdots,0,1,0,\cdots)$. The vertex corresponds to the outcome $E_j$ of this POVM $\{E_i\}^n_{i=1}$. An example of performing a 3-outcome POVM is shown in Fig.~\ref{walkplot}.

\begin{figure}
  \includegraphics[width=8cm]{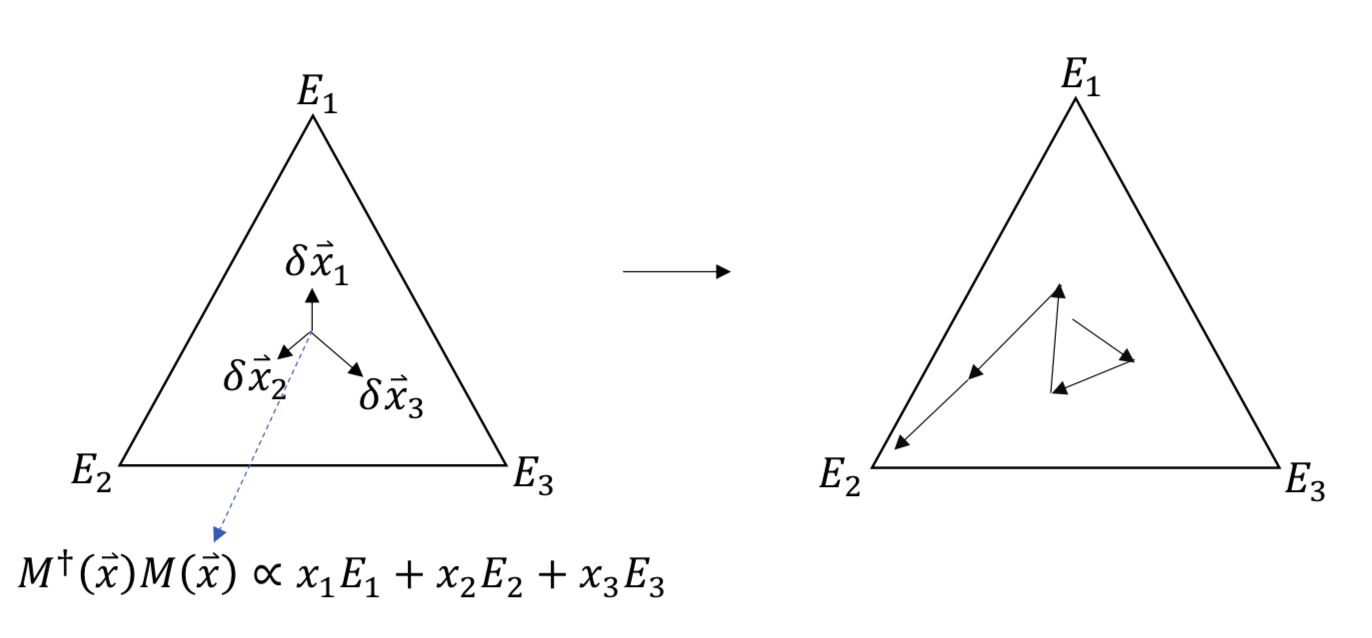}
  \caption{The vector $\vec{x}$, starting from the center, characterizes the position of the walk in $\triangle^{2}$. A 3-outcome measurement is performed at every $\vec{x}$. Each outcome corresponds to a small displacement that takes $\vec{x}$ to a new position $\vec{x}+\delta\vec{x}_k$. At long times, the walk approaches one of the vertices.}
  \label{walkplot}
\end{figure}

Before making any measurement, $\vec{x}$ is at the middle of $\triangle^{n-1}$, so the random walk starts from 
\begin{equation}
\vec{x}=\vec{x}_0=\frac{1}{n} \left(
\begin{array}{c}
1\\
\vdots\\
1
\end{array}
\right),
\end{equation}
and $M(\vec{x}_0)=U(\vec{x}_0)=I$. After each weak measurement is made, the new position $\vec{x}$ remains in $\triangle^{n-1}$ if the weak measurement, with measurement operators $\{M_{k}(\vec{x})\}^n_{k=1}$ satisfying
\begin{equation}
\sum^n_{k=1}M^{\dagger}_k (\vec{x}) M_k (\vec{x})=I,
\end{equation}
satisfies 
\begin{equation}
\label{walkcond1}
M^{\dagger}(\vec{x})M^{\dagger}_{k}(\vec{x})M_k (\vec{x})M(\vec{x})\propto M^{\dagger}(\vec{x}_k)M(\vec{x}_k)
\end{equation}
for each outcome $k$ such that
\begin{equation}
\vec{x}_k\in \triangle^{n-1}.
\end{equation}
Note that $\vec{x}_k$ is a new position in the simplex $\triangle^{n-1}$ after outcome $k$, and $\delta\vec{x}_k\equiv \vec{x}_k-\vec{x}$ is expected to be a small quantity given that we are performing weak measurements. 
Using Eq. (\ref{operatorform}) for $M(\vec{x})$, we can rewrite Eq. (\ref{walkcond1}) as
\begin{align}
\label{cond1}
&\sqrt{\sum^n_{i=1}x_i E_i}U^{\dagger}(\vec{x})M^{\dagger}_{k}(\vec{x})M_k (\vec{x})U(\vec{x})\sqrt{\sum^n_{i=1}x_i E_i} \nonumber \\
&\propto  \sum^n_{i=1} (\vec{x}_k)_i E_i.
\end{align}
Define a new set of POVM elements by
\begin{equation}
\label{newE}
E_i(\vec{x}) \equiv U(\vec{x}) E_i U^{\dagger}(\vec{x}), \ \ \forall \ i \in \{1,\cdots,n\}.
\end{equation}
It can be easily checked that $\{E_i\}^n_{i=1}$ is linearly independent if and only if $\{E_i(\vec{x})\}^n_{i=1}$ is linearly independent. Inserting Eq. (\ref{newE}) back into Eq. (\ref{cond1}), we have 
\begin{align}
\label{conditionx}
& M^{\dagger}_{k}(\vec{x})M_k (\vec{x})\propto \nonumber \\
&\left( \sum^n_{i=1} x_i E_i(\vec{x})\right)^{-\frac{1}{2}} \left( \sum^n_{i=1} (\vec{x}_k)_i E_i(\vec{x}) \right) \left( \sum^n_{i=1} x_i E_i(\vec{x})\right)^{-\frac{1}{2}}.
\end{align}

We can expand all the operators as linear combinations of Pauli operators and the identity; i.e., for each $i$,
\begin{equation}
\label{Epauli}
E_i(\vec{x})=q_i(I+\vec{v}_i \cdot \vec{\sigma}), 
\end{equation}
where $\vec{\sigma}$ is the vector of Pauli matrices and $q_i$, $\vec{v}_i$ are the expansion coefficients. By assumption, all $q_i$'s are strictly positive. For convenience, the indication of $\vec{x}$ dependence for $q_i$ and $\vec{v}_i$ have been dropped. We have
\begin{align}
\label{pauliexpand}
 \sum^n_{i=1} x_i E_i(\vec{x})&=  \sum^n_{i=1} x_i q_i I +  \sum^n_{i=1} x_i q_i\vec{v}_i \cdot \vec{\sigma} \nonumber \\
&\propto I+\frac{1}{\sum^n_{i=1}x_i q_i} \left(\sum^n_{i=1} x_i q_i \vec{v}_i\right)\cdot \vec{\sigma} \nonumber \\
& = I+ \vec{r} \cdot \vec{\sigma},
\end{align}
where 
\begin{equation}
\vec{r}\equiv \frac{1}{\sum^n_{i=1}x_i q_i} \left(\sum^n_{i=1} x_i q_i \vec{v}_i\right).
\end{equation}
Using the property shown in Appendix A, the inverse square root of Eq. (\ref{pauliexpand}) is
\begin{equation}
\label{sqrtE}
\left(\sum^n_{i=1} x_i E_i(\vec{x})\right)^{-\frac{1}{2}} \propto\left(I+\vec{r}\cdot \vec{\sigma}\right)^{-\frac{1}{2}}\propto (I-b \vec{r}\cdot \vec{\sigma}),
\end{equation}
where
\begin{equation}
\label{27}
b=\frac{1-\sqrt{1-|\vec{r}|^2}}{|\vec{r}|^2}.
\end{equation}
Similarly, we can write 
\begin{equation}
\label{pauliK}
\sum^n_{i=1} (\vec{x}_k)_i E_i(\vec{x}) \propto I+\vec{r}_k \cdot \vec{\sigma},
\end{equation}
where 
\begin{equation}
\vec{r}_k\equiv \frac{1}{\sum^n_{i=1}(\vec{x}_k)_i q_i} \left(\sum^n_{i=1} (\vec{x}_k)_i q_i \vec{v}_i\right),
\end{equation}
and define $\delta\vec{r}_k\equiv \vec{r}_k-\vec{r}$.
Inserting Eqs.~(\ref{sqrtE}) and (\ref{pauliK}) to Eq.~(\ref{conditionx}), the condition becomes
\begin{align}
M^{\dagger}_k(\vec{x})M_k(\vec{x})&\propto (I-b \vec{r}\cdot \vec{\sigma})\left(I+\vec{r}_k \cdot \vec{\sigma}\right)(I-b \vec{r}\cdot \vec{\sigma}) \label{30}\\
&\propto I+ \frac{b(\vec{r}\cdot \delta\vec{r}_k)\vec{r}+(\frac{1}{b}-1)\delta\vec{r}_k}{1-\vec{r}\cdot \vec{r}_k}\cdot \vec{\sigma}. \label{31}
\end{align}
The transformation from Eq.~(\ref{30}) to Eq.~(\ref{31}) is by the properties of Pauli matrices and $1-2b+b^2|\vec{r}|^2=0$ from Eq.~(\ref{27}). Hence,
\begin{align}
\label{equation32}
&M^{\dagger}_k(\vec{x})M_k(\vec{x})= \nonumber \\
&c_k\left[ \left(1-\vec{r}\cdot \vec{r}_k\right)I +\left(b(\vec{r}\cdot \delta\vec{r}_k)\vec{r}+(\frac{1}{b}-1)\delta\vec{r}_k\right)\cdot \vec{\sigma} \right],
\end{align}
where the $c_k$'s are $positive$ constants to be determined. The condition to be a measurement,
\begin{equation}
\sum^n_{k=1}M^{\dagger}_k(\vec{x})M_k(\vec{x})=I,
\end{equation}
requires 
\begin{eqnarray}
 \left\{ \begin{array}{l}
 \displaystyle\sum^n_{k=1}c_k(1-|\vec{r}|^2) - \displaystyle\sum^n_{k=1} c_k \delta\vec{r}_k \cdot \vec{r}=1,  \\
b\left(\vec{r}\cdot \displaystyle\sum^n_{k=1} c_k \delta\vec{r}_k\right)\vec{r}  + \left(\frac{1}{b}-1\right)\displaystyle\sum^n_{k=1}c_k \delta\vec{r}_k =0.
 \end{array} \right.
\end{eqnarray}
It is sufficient to solve the following conditions, 
\begin{eqnarray}
\label{condition35}
 \left\{ \begin{array}{l}
  \displaystyle\sum^n_{k=1}c_k=\frac{1}{1-|\vec{r}|^2}, \\
   \displaystyle\sum^n_{k=1}c_k \delta\vec{r}_k=0.
 \end{array} \right. 
\end{eqnarray}
Note that $|\vec{r}|$ cannot be 1 because $|\vec{r}|=1$ happens only if the $E_i$'s are $all$ proportional to the $same$ projector, but this contradicts the assumption that the $E_i$'s are linearly independent.

The question now becomes whether there exists a set of $c_k$'s ($> 0$) and $\vec{x}_k$'s ($\in \triangle^{n-1}$) such that Eq.~(\ref{condition35}) is satisfied. First of all, $\sum^n_{i=1}E_i(\vec{x})=I$ implies that
\begin{equation}
\label{qv}
\vec{q}=(q_1,\cdots,q_n)\in\triangle^{n-1} \ \ \text{and} \ \ \sum^n_{i=1}q_i \vec{v}_i=0. 
\end{equation}
The $\vec{v}_i$'s and $\vec{r}$ are 3-dimensional Bloch vectors and $\vec{x}$ is an n-dimensional vector in $\triangle^{n-1}$, where $n$ can only be 2, 3 or 4 in our case. It is shown in Appendix B that there is a one-to-one map between $\triangle^{n-1}$ and the ``allowed'' space of $\vec{r}$, denoted as $\triangle^{\vec{r}}$,
\begin{equation}
\triangle^{\vec{r}}\equiv \{\vec{r}\mid \vec{r}= \frac{1}{\vec{x}\cdot\vec{q}}\left(\sum^n_{i=1}x_i q_i \vec{v}_i\right), \ \forall\vec{x}\in\triangle^{n-1}\}.
\end{equation}
For $n=2$, $\triangle^1$ is a line and $\triangle^{\vec{r}}$ is a line spanned by a pair of parallel and opposite-directioned vectors $\vec{v}_{1,2}$; for $n=3$, $\triangle^2$ is a regular triangle and $\triangle^{\vec{r}}$ is a 2-D triangle stretched by $\vec{v}_{1,2,3}$; for $n=4$, $\triangle^3$ is a regular tetrahedron and $\triangle^{\vec{r}}$ is a tetrahedron stretched by $\vec{v}_{1,2,3,4}$. An example of the $n=3$ case is illustrated in Fig.~\ref{simplex_pic}. Note that the $\vec{v}_i$'s are $E_i(\vec{x})$'s Bloch vectors, which are rotated from the Bloch vectors of the original $E_i$'s.

It is clear that there are many solutions to Eq.~(\ref{condition35}); one simple solution is 
\begin{eqnarray}
\left\{ \begin{array}{l}
  c_k=\frac{1}{n(1-|\vec{r}|^2)}, \ \ \ \ \forall k\in \{1,\cdots,n\},\\
   \displaystyle\sum^n_{k=1}\delta\vec{r}_k=0,
 \end{array} \right.
\end{eqnarray}
where the $\delta\vec{r}_k=\vec{r}_k-\vec{r}$ are vectors lying in $\triangle^{\vec{r}}$. For example, when $n=3$, the three vectors $\delta\vec{r}_k$'s lie in a triangle and add up to zero vector. The three $\delta\vec{r}_k$'s are given by three $\vec{r}_k$'s, which correspond to three $\vec{x}_k$'s in $\triangle^{n-1}$. In fact, we can solve for the set of $\vec{x}_k$'s if and only if we can solve for the set of $\delta\vec{r}_k$'s. 
\begin{figure}
  \includegraphics[width=8cm]{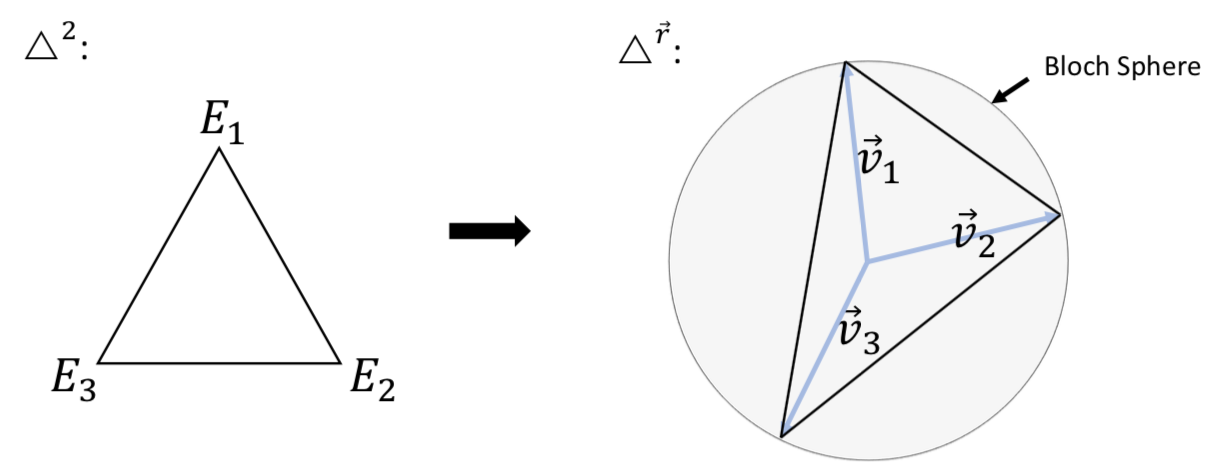}
  \caption{The map from $\triangle^{2}$ to a triangle $\triangle^{\vec{r}}$ in the Bloch sphere.}
  \label{simplex_pic}
\end{figure}

\subsection{Constraints by destructive weak measurements}
So far we have shown that it is possible to perform a weak measurement while staying in a simplex. However, it was assumed that we can choose arbitrary sets of $\vec{x}_k$'s, which correspond to arbitrary weak measurements. We now show that even using a more restricted class of weak measurements---destructive weak measurement---Eq.~(\ref{condition35}) still can be satisfied. Therefore we can perform a series of destructive weak measurements such that it corresponds to a random walk in a simplex. However, the choice of $\delta\vec{r}_k$'s in Eq.~(\ref{condition35}) are more restricted in this case.

In Sec.~\ref{0}, we reintroduced the model of destructive weak measurements in \cite{Chen:1}. It consists of coupling the system with an ancilla qubit in state $|0\rangle$ by a weak swap operation, followed by a projective POVM on the ancilla. The outcome $\hat{P}_k=s_k|e_k\rangle \langle e_k|$ of this projective POVM corresponds to a measurement operator 
\begin{eqnarray}
\label{meas_oper_}
M_k=\sqrt{s_k} \left(\begin{array}{cc}
\langle e_k|0\rangle e^{-i\phi}& -i \langle e_k|1\rangle\sin\phi\\
0& \langle e_k|0\rangle \cos\phi
\end{array}\right)
\end{eqnarray}
on the system. From the other direction, as shown in \cite{Chen:1}, for a set of eigenvectors $\{|b_k\rangle,|\overline{b_k}\rangle\}$ of $M^{\dagger}_kM_k$,
where
 \begin{equation}
|b_k\rangle= \alpha_k |0\rangle + \beta_k e^{i\chi_k}|1\rangle, \ \ |\overline{b_k}\rangle=\beta_k |0\rangle -\alpha_k e^{i\chi_k}|1\rangle, 
\end{equation}
and 
\begin{equation}
\label{eqn44}
\alpha_k,\beta_k\geq 0,\ \  \alpha_k^2+\beta^2_k=1, \  \ \chi_k\in \mathbb{R},
\end{equation}
there exists a corresponding vector $|e_k\rangle$ of $\hat{P}_k$ on the ancilla that will give us those eigenvectors. With this $|e_k\rangle$, the eigenvalues of $M^{\dagger}_kM_k$ take the form 
\begin{equation}
\lambda_k,\overline{\lambda}_k= \frac{1}{2}s_k\left(1+\frac{\sqrt{1-4\alpha_k^2\beta^2_k}\cos^2\phi\pm\sin^2\phi}{\sqrt{1-4\alpha^2_k\beta^2_k\cos^2\phi}}\right).
\end{equation}
On the other hand, one can check that if $\sum^n_{k=1}M^{\dagger}_kM_k=I$ then the corresponding $\hat{P}_k$'s on the ancilla form a projective POVM. 

This result can be translated into the Bloch sphere picture for $M^{\dagger}_kM_k$:
\begin{align}
M^{\dagger}_kM_k&=\lambda_k|b_k\rangle\langle b_k| +\overline{\lambda}_k |\overline{b_k}\rangle\langle\overline{b_k}| \nonumber \\
&= \frac{\left(\lambda_k+\overline{\lambda}_k\right)}{2}\left(I+\frac{\lambda_k-\overline{\lambda}_k}{\lambda_k+\overline{\lambda}_k}\hat{n}_k\cdot\vec{\sigma}\right),
\end{align}
where the unit vector $\hat{n}_k$ is
\begin{eqnarray}
\hat{n}_k=\left( \begin{array}{c}
2 \alpha_k \beta_k \cos\chi_k\\
-2\alpha_k \beta_k \sin\chi_k\\
\alpha^2_k-\beta^2_k
\end{array}
\right).
\end{eqnarray}
The explicit expression for $M^{\dagger}_kM_k$ is
\begin{align}
\label{eqn48}
&M^{\dagger}_kM_k= \frac{s_k}{2}\left(2+\frac{2|(\hat{n}_k)_z|\cos^2 \phi}{\sqrt{\sin^2\phi+(\hat{n}_k)_z^2 \cos^2\phi}}\right) \nonumber \\
&\times \left[I+\frac{\sin^2\phi}{|(\hat{n}_k)_z|\cos^2\phi+\sqrt{\sin^2\phi+(\hat{n}_k)^2_z\cos^2\phi}}\hat{n}_k\cdot \vec{\sigma} \right].
\end{align}
Now the question becomes: can we use the destructive weak measurements to perform a series of measurements which keep the evolution of POVM elements in the simplex? More specifically, for a given position $\vec{x}$ in the simplex, can we match Eq.~(\ref{equation32}) with Eq.~(\ref{eqn48}), for a set of $\{\delta\vec{r}_k\}$ and $\{c_k\}$ that satisfy Eq. (\ref{condition35}) and a set of $\{s_k >0\}$ and $\{\alpha_k,\beta_k\}$ that satisfy Eq.~(\ref{eqn44})? 

Note that in Eq.~(\ref{eqn48}) the vector $\hat{n}_k$ can be any unit vector, and the overall positive constant, 
\begin{equation}
\frac{s_k}{2}\left(2+\frac{2|(\hat{n}_k)_z|\cos^2 \phi}{\sqrt{\sin^2\phi+(\hat{n}_k)_z^2 \cos^2\phi}}\right),
\end{equation}
can be any positive number because $s_k$ can be freely chosen. The only constraint imposed by Eq.~(\ref{eqn48}) is that the Bloch vector's length,
\begin{equation}
\label{eqn50}
\frac{\sin^2\phi}{|(\hat{n}_k)_z|\cos^2\phi+\sqrt{\sin^2\phi+(\hat{n}_k)^2_z\cos^2\phi}},
\end{equation}
is determined by the z-component of $\hat{n}_k$. On the other hand, the requirement that the POVM stay in the simplex means that the measurement operators must satisfy Eq.~(\ref{equation32}), which can be rewritten as  
\begin{align}
\label{eqn51}
&M^{\dagger}_k(\vec{x})M_k(\vec{x})=c_k(1-|\vec{r}|^2-\vec{r}\cdot \delta\vec{r}_k) \nonumber \\
&\times \left[I+ \frac{|\delta\vec{r}_k|\left((\hat{r}\cdot \delta\hat{r}_k)\hat{r}+\sqrt{1-|\vec{r}|^2}(\hat{r}_{\perp}\cdot \delta\hat{r}_k)\hat{r}_{\perp}\right)}{1-|\vec{r}|^2-\vec{r}\cdot \delta\vec{r}_k}\cdot \vec{\sigma} \right],
\end{align}
where $\hat{r},\delta\hat{r}_k$ are the unit vectors of $\vec{r},\delta\vec{r}_k$, and
\begin{equation}
\delta\hat{r}_k=(\hat{r}\cdot \delta\hat{r}_k)\hat{r}+(\hat{r}_{\perp}\cdot \delta\hat{r}_k)\hat{r}_{\perp},\ \ \  \ \ \hat{r}_{\perp}\cdot\hat{r}=0.
\end{equation}
The term $(\hat{r}_{\perp}\cdot \delta\hat{r}_k)\hat{r}_{\perp}$ is just the orthogonal component of $\delta\hat{r}_k$ in the direction $\hat{r}_{\perp}$. (We use hats $\hat{v}$ to denote unit vectors.) The constraint is that the length of the Bloch vector of Eq.~(\ref{eqn51}) must equal Eq.~(\ref{eqn50}),
\begin{align}
\label{eqn53}
& \frac{|\delta\vec{r}_k|\sqrt{1-|\vec{r}|^2\sin^2\theta_k}}{1-|\vec{r}|^2-|\vec{r}||\delta\vec{r}_k|\cos\theta_k} \nonumber \\
&=\frac{\sin^2\phi}{|(\hat{n}_k)_z|\cos^2\phi+\sqrt{\sin^2\phi+(\hat{n}_k)^2_z\cos^2\phi}},
\end{align}
where $\cos\theta_k\equiv \hat{r}\cdot \delta\hat{r}_k$ and 
\begin{equation}
(\hat{n}_k)_z=\frac{(1-\sqrt{1-|\vec{r}|^2})(\hat{r})_z\cos\theta_k+\sqrt{1-|\vec{r}|^2}(\delta\hat{r}_k)_z}{\sqrt{1-|\vec{r}|^2\sin^2\theta_k}}.
\end{equation}
Note that the right hand side of Eq.~(\ref{eqn53}) depends only on the direction of $\delta\vec{r}_k$ (which is $\delta\hat{r}_k$), and it is a number between 0 and $\sin\phi$. We can solve Eq.~(\ref{eqn53}) by first determining the direction of $\delta\vec{r}_k$ and then solving for the length of $\delta\vec{r}_k$. For a given current position $\vec{r}$ in the simplex $\triangle^{\vec{r}}$, we first choose the directions of $\delta\vec{r}_k$'s by pointing from $\vec{r}$ toward the vertices of $\triangle^{\vec{r}}$; i.e., for each outcome $k$, the direction of $\delta\vec{r}_k$ is chosen to be parallel to $\vec{v}_k-\vec{r}$. Therefore,
\begin{equation}
\cos\theta_k=\hat{r}\cdot\delta\hat{r}_k=\hat{r}\cdot\frac{\vec{v}_k-\vec{r}}{\left| \vec{v}_k-\vec{r}\right|}.
\end{equation}
Inserting these directions back into Eq.~(\ref{eqn53}), the left hand side becomes an increasing function of $|\delta\vec{r}_k|$: it increases from 0 to 1 as $|\delta\vec{r}_k|$ goes from 0 to $\left| \vec{v}_k-\vec{r}\right|$. For the directions $\delta\hat{r}_k$ chosen above, we can aways find lengths $|\delta\vec{r}_k|$ to satisfy Eq.~(\ref{eqn53}). For such $\delta\vec{r}_k$'s, the $c_k$'s can always be chosen to satisfy Eq.~(\ref{condition35}) because the $\delta\vec{r}_k$'s are coplanar. Therefore, we have a set of $\delta\vec{r}_k$'s (and hence $\vec{x}_k$'s) and $c_k$'s that satisfy the constraints of the model of destructive weak measurement and the requirements to perform a weak measurement that corresponds to a random walk in a simplex.

\section{The probability to obtain an outcome}
\label{3}

So far we have shown that one can perform, by the model of destructive weak measurements, a series of weak measurements which keep the evolution of the corresponding positive operator in a simplex. In this section, we show that, at long times, the walk in the simplex will approach one of the vertices. Each vertex corresponds to an outcome of this series of weak measurements. The probability of obtaining an outcome $k$ from the series of weak measurements approaches the probability of obtaining $E_k$ if we performed the original POVM, $\{E_i\}^n_{i=1}$.

Note that we only have to consider the case of linearly independent and projective POVMs (LIPPOVM) as discussed at the end of Sec.~\ref{1}. The following subsection will assume the POVM is a LIPPOVM.

\subsection{Approaching a vertex}
From the form of the measurement operator in Eq. (\ref{meas_oper_}), it is evident that after performing $N$ steps, the measurement operators accumulate as
\begin{align}
M_{k_N}\cdots M_{k_1}&\propto \left(\begin{array}{cc}
e^{-iN\phi} & \kappa \\
0 & \cos^N\phi
\end{array}\right)  \\
&=|0\rangle \langle \xi|+\cos^N\phi|1\rangle\langle 1|,
\end{align}
where $\kappa$ is a number depending on the string of outcomes $k_N\cdots k_1$ and 
\begin{eqnarray}
| \xi\rangle \equiv \left(\begin{array}{c}
e^{iN\phi}\\
\kappa^{*}
\end{array}
\right).
\end{eqnarray}
The POVM element of this string of outcomes is 
\begin{align}
E_{k_{N}\cdots k_1}&=M^{\dagger}_{k_1}\cdots M^{\dagger}_{k_N}M_{k_N}\cdots M_{k_1} \nonumber \\
&\propto |\xi\rangle \langle \xi| + \cos^{2N}\phi |1\rangle \langle1|.
\end{align}
Thus for large $N$, the POVM element approaches a projector, $|\xi\rangle \langle \xi|$, at a rate $\cos^{2N}\phi$. On the other hand, because we only consider LIPPOVMs, the simplex, $\triangle^{\vec{r}}$, has vertices on the surface of the Bloch sphere. Those vertices represent the LIPPOVM's elements which are proportional to projectors. The one-to-one correspondence between $\triangle^{\vec{r}}$ and $\triangle^{n-1}$ implies that their vertices are mapped to each other. Recall from Eq.~(\ref{POVMelem}) that when $\vec{x}$ approaches one of the vertices, the walk approaches a projector, which is one of the elements of the LIPPOVM. Therefore, in the large $N$ limit, the random walk on the simplex given by the destructive weak measurements must approach a vertex.

\subsection{The probability for each outcome}
After $N$ steps, we denote the position in the simplex as $\vec{x}^{(i)}_{k_N\cdots k_1}$ if it is close to the vertex $E_i$. The index $k_\nu$ indicates the outcome of the $\nu$th weak measurement, and the $\nu$th measurement depends on all previous outcomes $k_{\nu-1}\cdots k_{1}$. Recall from Eq.~(\ref{POVMelem}), that if the walk approaches the vertex $E_1$, for example, after $N$ steps, the POVM element becomes
\begin{eqnarray}
E(\vec{x}^{(1)}_{k_N\cdots k_1}) = C(\vec{x}^{(1)}_{k_N\cdots k_1})\left[\sum^n_{i=1}(\vec{x}^{(1)}_{k_N\cdots k_1})_i E_i\right], \label{E1prob}
\end{eqnarray}
 where 
 \begin{equation}
(\vec{x}^{(1)}_{k_N\cdots k_1})_1\to1,\ \ \ (\vec{x}^{(1)}_{k_N\cdots k_1})_i\to 0\ \ \text{for}\ i=2,\cdots,n,
\end{equation}
and $C(\vec{x}^{(1)}_{k_N\cdots k_1})$ is a constant depending on the string of outcomes.
The probability of obtaining the outcome ``1" for a state $\rho$ is
\begin{equation}
\label{p1}
p_1\equiv\sum_{k_N\cdots k_1\in \text{``1''}}\Tr \left[ \rho E(\vec{x}^{(1)}_{k_N\cdots k_1})\right],
\end{equation}
where the summation is over all strings of outcomes $k_N\cdots k_1$ such that the walk approaches the vertex $E_1$. To show that $p_1$ indeed approaches the probability given by the usual Born's rule, $\Tr\left[ \rho E_1\right]$, it is equivalent to show that
\begin{equation}
\hat{p}_1\equiv\sum_{k_N\cdots k_1\in \text{``1''}} E(\vec{x}^{(1)}_{k_N\cdots k_1})\to E_1.
\end{equation}

We first claim that the sum of the constants over the strings of outcomes,
\begin{equation}
\sum_{k_N\cdots k_1\in \text{``1''}}C(\vec{x}^{(1)}_{k_N\cdots k_1}), 
\end{equation}
is bounded. Indeed, we can find a unit vector $|\zeta\rangle$ such that for all $ i\in\{1,\cdots,n\}$, $\langle \zeta| E_i |\zeta\rangle\neq 0$. Such a $|\zeta\rangle$ exists because $\{E_i\}^n_{i=1}$ is a qubit LIPPOVM. From Eqs.~(\ref{E1prob}) and (\ref{p1}), we have 
\begin{align}
&\text{min}\{\langle \zeta| E_i |\zeta\rangle\}\times\sum_{k_N\cdots k_1\in \text{``1''}}C(\vec{x}^{(1)}_{k_N\cdots k_1})\nonumber \\
&\leq\langle \zeta| \hat{p}_1 |\zeta\rangle \leq\langle \zeta| \sum^n_{i=1}\hat{p}_i |\zeta\rangle=\langle \zeta| I |\zeta\rangle=1.
\end{align}
Hence, it is bounded by $1/\text{min}\{\langle \zeta| E_i |\zeta\rangle\}$. In the large $N$ limit, for each $i\in\{2,\cdots,n\}$, $(\vec{x}^{(1)}_{k_N\cdots k_1})_i\to0$ for all strings of outcomes $k_N\cdots k_1$. Therefore, we have
\begin{align}
&\sum_{k_N\cdots k_1\in \text{``1''}}C(\vec{x}^{(1)}_{k_N\cdots k_1})(\vec{x}^{(1)}_{k_N\cdots k_1})_i \nonumber\\
&\leq \text{max}\{(\vec{x}^{(1)}_{k_N\cdots k_1})_i\}\sum_{k_N\cdots k_1\in \text{``1''}}C(\vec{x}^{(1)}_{k_N\cdots k_1}) \nonumber\\
&\leq \frac{ \text{max}\{(\vec{x}^{(1)}_{k_N\cdots k_1})_i\}}{\text{min}\{\langle \zeta| E_i |\zeta\rangle\}}\to 0,
\end{align}
where $\text{max}\{(\vec{x}^{(1)}_{k_N\cdots k_1})_i\}$ is the maximum of $(\vec{x}^{(1)}_{k_N\cdots k_1})_i$ among all strings of outcomes. The coefficients, $\sum_{k_N\cdots k_1\in \text{``1''}}C(\vec{x}^{(1)}_{k_N\cdots k_1})(\vec{x}^{(1)}_{k_N\cdots k_1})_i$, approach zero for $i\neq1$. For the same reason, other outcomes won't contribute to ``1'' either. Because $\sum^n_{i=1}\hat{p}_i=I=\sum^n_{i=1}E_i$ and the $\{E_i\}^n_{i=1}$ are linearly independent, we have 
\begin{align}
&\displaystyle\sum_{k_N\cdots k_1\in \text{``1''}}C(\vec{x}^{(1)}_{k_N\cdots k_1})(\vec{x}^{(1)}_{k_N\cdots k_1})_i\to \delta_{1i}, \nonumber \\
&\ \ \ \ \ \ \ \ \ \ \ \ \ \ \ \ \ \ \ \ \ \ \ \ \ \vdots \ \ \ \nonumber \\
&\displaystyle\sum_{k_N\cdots k_1\in \text{``n''}}C(\vec{x}^{(n)}_{k_N\cdots k_1})(\vec{x}^{(n)}_{k_N\cdots k_1})_i\to \delta_{ni}.
\end{align}
This verifies that $p_i$ does approach $\Tr\left[\rho E_i\right]$, and the POVM element finally evolves close to a vertex $E_i$ with probability approaching $\Tr\left[\rho E_i\right]$.

\section{Conclusion}\label{conclusion}

The model of destructive weak measurement in \cite{Chen:1} captures the behavior of a destructive strong measurement that leaves the system in a fixed final state after the measurement regardless of the outcome.  The only freedom in this model is the ability to choose any projective positive-operator-valued measurement on the ancilla in an adaptive way. Surprisingly,  \cite{Chen:1} shows that, in spite of its limitations, this model can achieve any projective measurement for qubits, and the result can be generalized to any positive-operator-valued measurement (POVM) with commuting POVM elements. However, the method used in \cite{Chen:1} does not work for POVMs with non-commuting elements.

In this paper we have shown that {\it any} qubit POVM can be achieved by this model. The improvement from POVMs with commuting elements to all POVMs requires a more complicated generalization of the previous method. The approach combines classical pre- and post-processing with a continuous measurement that corresponds to a random walk in a simplex. The pre-processing randomly chooses a linearly independent POVM with elements proportional to projectors from the original POVM. This linearly independent POVM is decomposed into a sequence of destructive weak measurements by mapping the evolution of the positive operator to a random walk in a simplex. The vertex the walk approaches corresponds to the final outcome of this succession of weak measurements.  Finally, the post-processing chooses which result to output, conditioned on the outcome of the quantum measurement process. The probability of each result agrees with the probability of that result in the original POVM. 

We have shown that a limited model of destructive weak measurements can achieve an arbitrary POVM for qubits. In higher dimensional systems, what measurements can be achieved by such a restricted class of weak measurements is an open question. The scheme of a system interacting with a stream of qubit probes was shown to be restricted to provide measurements with two distinct singular values  \cite{Florjanczyk:2014aa}. Generalizing the ancilla to higher dimensions gives more degrees of freedom, and that may be a possible route to achieve POVMs for higher dimensional systems. This is the subject of ongoing research.

\section*{Acknowledgments} 

YHC and TAB acknowledge useful conversations with Christopher Cantwell, Ivan Deutsch, Shengshi Pang, Christopher Sutherland, and Howard Wiseman. This work was funded in part by the ARO MURI under Grant No. W911NF-11-1-0268, NSF Grant No. CCF-1421078, NSF Grant No. QIS-1719778, and by an IBM Einstein Fellowship at the Institute for Advanced Study.

\appendix
\section{$E^{-1/2}$}
For a positive-semidefinite 2-by-2 matrix $E$, we can express $E$ in terms of Pauli matrices and the identity:
\begin{equation}
E=\frac{a}{2}\left(I+\vec{r}\cdot \vec{\sigma}\right) \ \ \ \ \ |\vec{r}|\leq 1, \ \ a>0.
\end{equation}
Decomposing it into its eigenvectors and eigenvalues, we have 
\begin{align}
E&=\lambda_{+}|\psi_{+}\rangle \langle \psi_{+}| +\lambda_{-}|\psi_{-}\rangle \langle \psi_{-}| \nonumber \\
&= \frac{1}{2}\left[(\lambda_{+} + \lambda_{-})I+(\lambda_{+} -\lambda_{-})\frac{\vec{r}}{|\vec{r}|}\cdot \vec{\sigma}\right],
\end{align}
where
\begin{equation}
|\psi_{\pm}\rangle \langle \psi_{\pm}|=\frac{1}{2}\left(I\pm \frac{\vec{r}}{|\vec{r}|}\cdot \vec{\sigma}\right)
\end{equation}
and 
\begin{equation}
\lambda_{\pm}=\frac{a}{2}(1\pm |\vec{r}|).
\end{equation}
The inverse square root is
\begin{align}
E^{-\frac{1}{2}}&= \frac{1}{\sqrt{\lambda_+}}|\psi_{+}\rangle \langle \psi_{+}| +\frac{1}{\sqrt{\lambda_-}}|\psi_{-}\rangle \langle \psi_{-}| \\
&= \frac{1}{2}\sqrt{\frac{2}{a}}\left(\frac{\sqrt{1+|\vec{r}|}+\sqrt{1-|\vec{r}|}}{\sqrt{1-|\vec{r}|^2}}\right) \\
&\ \ \ \times \left(I-\frac{\sqrt{1+|\vec{r}|}-\sqrt{1-|\vec{r}|}}{|\vec{r}|(\sqrt{1+|\vec{r}|}+\sqrt{1-|\vec{r}|})}\vec{r}\cdot \vec{\sigma}\right) \\
&\propto I-b \vec{r} \cdot \vec{\sigma},
\end{align}
where
\begin{equation}
b=\frac{1-\sqrt{1-|\vec{r}|^2}}{|\vec{r}|^2}.
\end{equation}

\section{Correspondence between $\triangle^{\vec{r}}$ and $\triangle^{n-1}$}
The picture relating the position indicator $\vec{x}\in\triangle^{n-1}$ and its 3-D image $\vec{r}$ in Bloch sphere is presented in this section. The $n$ can only be 2, 3 or 4 in our cases.

Recall from Eqs.~(\ref{Epauli}) and (\ref{pauliexpand}), the definition of $\vec{r}$ is 
\begin{equation}
\vec{r}=\frac{1}{\vec{x}\cdot \vec{q}}\sum^n_{i=1} x_i q_i \vec{v}_i,
\end{equation}
where $\vec{x},\vec{q}\in\triangle^{n-1}$, $q_i>0$ for all $i$, and the $\vec{v}_i$'s are the Bloch vectors corresponding to a set of linearly independent POVM elements. 

We define a bijection $T_{\vec{q}}:\triangle^{n-1}\to \triangle^{n-1}$ 
\begin{eqnarray}
T_{\vec{q}}(\vec{x})\equiv \frac{1}{\vec{x}\cdot \vec{q}}\left( \begin{array}{c}
q_1 x_1 \\
\vdots\\
q_n x_n
\end{array}\right),
\end{eqnarray}
for any $\vec{q}\in\triangle^{n-1}$, where $q_i>0$ for all $i=1,\cdots,n$.
The inverse map is given by 
\begin{eqnarray}
T^{-1}_{\vec{q}}(\vec{x})\equiv \frac{1}{\sum^n_{i=1} x_i q^{-1}_i}\left( \begin{array}{c}
q^{-1}_1 x_1 \\
\vdots\\
q^{-1}_n x_n
\end{array}\right).
\end{eqnarray}
We can define a vector $\vec{x}_{\vec{q}}\equiv T_{\vec{q}}(\vec{x})\in\triangle^{n-1}$ associated with the original $\vec{x}$. It is equivalent to analyze with $\vec{x}$ and with $\vec{x}_{\vec{q}}$ because of the bijection $T_{\vec{q}}$, i.e., we have a solution for $\vec{x}_{\vec{q}}$ if and only if we have a solution for $\vec{x}$. We choose to work with $\vec{x}_{\vec{q}}$ since it is linearly related to $\vec{r}$:
\begin{eqnarray}
\vec{r}=\left( \begin{array}{ccc}
 & & \\
\vec{v}_1& \cdots &\vec{v}_n\\
& & 
\end{array}\right) \frac{1}{\vec{x}\cdot \vec{q}}\left( \begin{array}{c}
q_1 x_1 \\
\vdots\\
q_n x_n
\end{array}\right)=V \vec{x}_{\vec{q}},
\end{eqnarray}
where
\begin{eqnarray}
V\equiv \left( \begin{array}{ccc}
 & & \\
\vec{v}_1& \cdots &\vec{v}_n\\
& & 
\end{array}\right)_{3 \times n}
\end{eqnarray}
is a 3 by $n$ matrix. The space $\vec{r}$ lives in is 
\begin{equation}
\triangle^{\vec{r}}=\left\{\vec{r}\ | \vec{r}=V\vec{x}_{\vec{q}}, \ \ \forall\vec{x}_{\vec{q}}\in\triangle^{n-1} \right\},
\end{equation}
which is a subspace spanned by $\{\vec{v}_1,\cdots,\vec{v}_n\}$. 

We claim that $\text{rank}(V)=n-1$ for $n=2,3,4$. The proof is as follows. Since $\sum^n_{i=1}q_i \vec{v}_i=0$ by Eq.~(\ref{qv}), we have $\vec{q}\in\text{Ker}(V)$. And by the linear independence of the set of POVM elements, 
\begin{equation}
E_i(\vec{x})=q_i(I+\vec{v}_i \cdot \vec{\sigma}),
\end{equation}
we have that for any vector $\vec{c}=(c_1,\cdots ,c_n)$ such that 
\begin{equation}
\sum^n_{i=1}c_i E_i(\vec{x})=0,
\end{equation}
then $\vec{c}=0$. This implies that there is no nontrivial solution $\vec{c}$ satisfying both
\begin{equation}
\label{condc}
\vec{c}\cdot \vec{q}=0\ \ \ \text{and}\ \ \ \sum^n_{i=1}c_i q_i \vec{v}_i=0.
\end{equation}
We can decompose the $n$-dimensional vector space into $\text{Span}\{\vec{q}\}\oplus\text{Span}\{\vec{q}\}^{\perp}$. Suppose a nontrivial vector $\vec{y}\in\text{Span}\{\vec{q}\}^{\perp}$ is also in $\text{Ker}(V)$, i.e.,
\begin{equation}
\vec{y}\cdot \vec{q}= 0\ \ \ \text{and}\ \ \ \sum^n_{i=1}y_i \vec{v}_i=V\vec{y}=0.
\end{equation}
Then there exists a vector 
\begin{eqnarray}
\vec{c}\equiv \left( \begin{array}{c}
\frac{y_1}{q_1} \\
\vdots \\
\frac{y_n}{q_n} 
\end{array}
\right) -\left(\sum^{n}_{i=1}y_i\right) \left( \begin{array}{c}
1\\
\vdots \\
1
\end{array}
\right)
\end{eqnarray}
satisfying both conditions in Eq.~(\ref{condc}), and one can easily check that $\vec{c}\neq 0$. This contradicts the fact that there is no nontrivial solution for Eq.~(\ref{condc}), and hence no vector in $\text{Span}\{\vec{q}\}^{\perp}$ is also in $\text{Ker}(V)$. We have $\text{Span}\{\vec{q}\}=\text{Ker}(V)$ and therefore $\text{dim}(\text{Ker}(V))=1$. By the rank-nullity theorem, we have $\text{rank}(V)=n-1$. 

Finally, to show that the map $V:\triangle^{n-1}\to \triangle^{\vec{r}}$ is one-to-one, we use the following argument. If $\vec{x}_{1,2}\in\triangle^{n-1}$ such that $V\vec{x}_1=V\vec{x}_2$, then $\vec{x}_1-\vec{x}_2= a \vec{q} \in\text{Ker}(V)$ for a number $a$. By the fact that $\vec{x}_{1,2}$ and $\vec{q}$ are in $\triangle^{n-1}$, $a$ must be zero, and  hence $\vec{x}_1=\vec{x}_2$. The map $V:\triangle^{n-1}\to \triangle^{\vec{r}}$ is onto, by the definition of $\triangle^{\vec{r}}$. Hence $\triangle^{n-1}$ and $\triangle^{\vec{r}}$ are isomorphic. We have $\text{dim}(\triangle^{\vec{r}}) = \text{dim}(\triangle^{n-1})=n-1$, and for $n=2,3, 4$, they are a line, a triangle, a tetrahedron.

\bibliography{QubitPOVM_v4c}

\end{document}